\documentclass{article}
\usepackage{graphicx} 
\usepackage{fullpage}

\usepackage{amsmath,amssymb,amsfonts,amsthm}
\newcommand{\F}{\mathbb{F}}
\DeclareMathOperator{\rank}{rank}

\usepackage[colorlinks]{hyperref}
\usepackage[capitalize,nameinlink]{cleveref}

\hypersetup{
  linkcolor=[rgb]{0,0,0.4},
  citecolor=[rgb]{0, 0.4, 0},
  urlcolor=[rgb]{0.6, 0, 0}
}

\newtheorem{theorem}{Theorem}[section]
\newtheorem{lemma}[theorem]{Lemma}
\newtheorem{claim}[theorem]{Claim}
\newtheorem{corollary}[theorem]{Corollary}

\newcommand{\calS}{\mathcal{S}}
\newcommand{\Q}{\mathbb{Q}}

\newcommand{\set}[1]{\left\{ #1 \right\}}

\title{Improved Algorithms for the Remote Point Problem}
\author{Ben Lee Volk\thanks{Efi Arazi School of Computer Science, Reichman University, Israel. Email: \texttt{benleevolk@gmail.com}. The research leading to these results has received funding from the  Israel Science Foundation (grant number 843/23). }}
\date{}

\begin{document}

\maketitle

\begin{abstract}
 The Remote Point Problem (RPP) is an algorithmic problem that asks, given a linear subspace $L \subseteq \F^n$ of dimension $k$, to deterministically find a vector $v \in \F^n$ far in Hamming distance from $L$. This problem was introduced by Alon, Panigrahy and Yekhanin \cite{APY09}, motivated in part by the matrix rigidity approach for proving circuit lower bounds. An algorithm is said to achieve remoteness $d$ if it finds a vector $v$ whose Hamming distance from $L$ is at least $d$.
 
We observe that over the rational numbers, the problem admits a deterministic polynomial-time algorithm that achieves optimal remoteness $n-k$.

Over finite fields, we obtain a (modest) improvement of a result of Alon, Panigrahy and Yekhanin \cite{APY09}, and give an algorithm that achieves remoteness $\Omega\left(\frac{n}{\max\{k, \log n\}} \log n\right)$.

\end{abstract}

\section{Introduction}

The \emph{Remote Point Problem} (RPP) over a field $\F$ is the following natural algorithmic problem: given (the basis of) a linear subspace $L \subseteq \F^n$ of dimension $k$ and a natural number $d$, find a vector $v \in \F^n$ which is $d$-far, in Hamming distance, from all points in $L$.

One motivation for studying this problem comes from the matrix rigidity approach for proving algebraic circuit lower bounds: Valiant \cite{Valiant77} defined a matrix $M \in \F^{n \times n}$ to be $(r,s)$-rigid if it cannot be written as a sum $M=R+S$ of a matrix $R$ of rank at most $r$ and a matrix $S$ with $s$-sparse rows.\footnote{Sometimes it is more convenient to bound the total sparsity of $S$, rather than the sparsity per-row. However, the differences between the two definitions are quite inconsequential, and for the purpose of proving circuit lower bounds it is enough to consider matrices in which every row is sparse, which is the definition that is more convenient for us.} In other words, for every $r$-dimensional subspace $L$, $M$ contains a row that is $(s+1)$-far (in Hamming distance) from $L$.  Thus, RPP can be seen as a ``white-box'' version of the problem: instead of constructing a set of vectors that work for all subspaces, the algorithm gets to ``see'' the subspace $L$ and to construct a bespoke $d$-far vector that may depend on $L$.

Random matrices are rigid with very good parameters, and Valiant has shown that explicit constructions of rigid matrices with good enough parameters will imply new circuit lower bounds. See \cite{Lokam09, Ramya20} for surveys on this topic. Similarly, in the RPP setting, a counting argument implies that as long as the dimension of $L$ is not too large, most vectors are indeed pretty far from $L$. For concreteness, consider an $n/2$-dimensional space $L \subseteq \F_2^n$: $L$ has only $2^{n/2}$ vectors, so a random vector $v \in \F_2^n$ has distance $\Omega(n)$ from $L$. The problem is to construct such a vector $v$ \emph{deterministically}. Hence, it is also a natural question in the theory of algebraic pseudorandomness.

RPP was introduced by Alon, Panigrahy and Yekhanin \cite{APY09}, who also gave a deterministic polynomial-time algorithm that finds an $\Omega(\log n)$-far vector from $n/2$-dimensional subspaces, or more generally, an $(n \log k / k)$-far vector from $k$-dimensional subspaces. Their algorithm is phrased over $\F_2$, but can be extended to any field.

Arvind and Srinivasan discovered more connections between RPP and circuit lower bounds \cite{AS10-lb}, and further generalized the algorithm of Alon, Panigrahy and Yekhanin \cite{APY09} to finite groups, and obtained a parallel polynomial-time algorithm (with similar remoteness guarantees) for abelian groups \cite{AS10-epsbiased}.

In this paper we improve those algorithms in several settings. First, we consider RPP over the field of rational numbers $\Q$. This problem turns out to have an optimal and very simple algorithm that relies mostly on elementary linear algebra.

\begin{theorem}
\label[theorem]{thm:rpp-over-Q}
There exists a deterministic polynomial-time algorithm that, given a basis of a $k$-dimensional subspace $L \subseteq \Q^n$, finds a vector $v$ which is $(n-k)$-far from $L$.
\end{theorem}

We remark that $(n-k)$ is the optimal remoteness parameter one could expect for a $k$-dimensional subspace. This follows for example from \cref{lem:syndrome-characterization} that we prove in \cref{sec:syndrome}. While the magnitude of the coordinates of the vector $v$ may be exponential in $n$, their bit complexity is merely polynomial in $n$ (and in the bit complexity of the given basis for $L$). It is interesting, of course, to obtain a vector with smaller entries satisfying the remoteness property, and this might be helpful in making progress for RPP over finite fields.

We find it quite surprising that such a simple and optimal algorithm exists for RPP over $\Q$. While, to the best of our knowledge, RPP over $\Q$ was not explicitly considered in the literature before, constructing rigid matrices over $\Q$ is a well-known and equally interesting open problem as in the case of finite fields (we refer again to the surveys \cite{Lokam09, Ramya20} for more background). It is therefore interesting that RPP seems to behave genuinely differently in this setting.

We also obtain a modest improvement to the algorithm of \cite{APY09}, when the dimension $k$ is subpolynomial in $n$.

\begin{theorem}
\label[theorem]{thm:intro:small-improvement}
Let $\F$ be a finite field and $L \subseteq \F^n$ be a linear subspace of dimension $k \le n/2$. There exists an algorithm that runs in time polynomial in $n$ and finds a point that is $\Omega(\frac{n}{\max\{k, \log n\}} \cdot \log n)$ far from $L$.
\end{theorem}

The assumption $k \le n/2$ is not crucial. A similar theorem can be derived whenever $k \le (1-\varepsilon)n$ for a constant $\varepsilon>0$. This will only affect the constant hidden under the big $\Omega$.

In \cite{APY09}, a similar theorem was proved with remoteness $\Omega(\frac{n}{k} \cdot \log k)$. When $k$ is subpolynomial in $n$, \cref{thm:small-improvement} gives an improvement.
For example, when $k=\log^C n$, the algorithm of \cite{APY09} gives remoteness of $\frac{n}{k} \log k$ whereas we can get remoteness $\frac{n}{k} k^{1/C}$.

\subsection{Technique}
\label[section]{sec:technique}

As mentioned above, the proof of \cref{thm:rpp-over-Q} is remarkably and surprisingly simple. We first observe an equivalent formulation for remoteness from a subspace, by considering what is known in coding theory as the \emph{syndrome} of the vector $v$, which is the vector $Hv \in \Q^{m}$ where $m=n-k$ and $H \in \Q^{m\times n}$ is a parity-check matrix for $L$, that is, $\ker H = L$. The vector $v$ is $d$-close to $L$ if and only if its syndrome $s$ can be written as a linear combination of at most $d$ columns of $H$. It follows that in order to solve RPP it is enough to find a syndrome vector $s$ that cannot be written as a linear combination of $d-1$ columns of $H$ (this is proved in \cref{sec:syndrome}). The description so far is in fact true over any field.

Over $\Q$, we find such $s$ simply by defining $s=(1,K,K^2, \ldots, K^{m-1})$ where $K$ is a certain (exponentially large) quantity that depends on the entries of $H$. We wish to show that $s$ is not in the image of any $m \times (m-1)$ submatrix $A$ of $H$. Elementary linear algebra guarantees the existence of a vector $a \in \Q^m$ with ``small'' entries such that $a^T A = 0$. In particular, the entries of $a$ in absolute value are much smaller than $K$, which guarantees that $a^T s \neq 0$. This implies that $s$ is not in the column span of $A$. The proof appears in \cref{sec:rpp-Q}.

\cref{thm:intro:small-improvement} is proved using a completely different technique. The main ingredient is an algorithm of \cite{APY09} that finds a vector $\Omega(\log n)$-far from subspaces of dimension $n/2$. To handle $k$-dimensional subspaces, \cite{APY09} partition the $n$ coordinates into $n/2k$ blocks of size $2k$ each. The projection of a $k$-dimensional subspace to each block (of size $2k$) has dimension $k$, which enables them to find a vector $v_i$ that is $\Omega(\log k)$-far on the $i$-th block. Concatenating the $n/2k$ vectors $v_1,v_2,\ldots,v_{n/2k}$ gives a vector of length $n$ that is $\Omega(n \log k / k)$-far.

We observe that since the algorithm is allowed to be polynomial in $n$, we can use the basic algorithm of \cite{APY09} with slightly different parameters on each block. Suppose for simplicity that $k \ge C \log n$ for some large enough constant $C$. We use the algorithm of \cite{APY09} to find a vector $v_i$ that is $\Omega(\log n)$-far from the projection of $L$ on the $i$-th block (instead of $\Omega(\log k)$): the distance is therefore larger than logarithmic in the ambient dimension $2k$. However, since we aim for an algorithm that is polynomial in $n$ (rather than polynomial in $k$), this still fits within the polynomial time budget. The details appear in \cref{sec:rpp-finite}.

\section{Syndrome-Based Characterization of Remote Points}
\label[section]{sec:syndrome}

We begin by showing the following elementary lemma, which relates the distance of $v$ from a subspace $L$ to the expressibility of the \emph{syndrome} of $v$, the vector $Hv \in \F^{n-k}$ (where $H$ is the parity-check matrix of $L$) as a linear combination of the columns of $H$. The results in this section hold over any field.

\begin{lemma}
\label[lemma]{lem:syndrome-characterization}
    Let $L$ be a linear subspace of dimension $k$, and let $H$ be a parity-check matrix for $L$, i.e., an $(n-k)\times n$ matrix of full row rank such that $L=\ker H$. Let $v \in \F^n$ and $s=Hv \in \F^{n-k}$. Then $v$ is $d$-close to $L$ if and only if $s$ can be written as a linear combination of at most $d$ columns of $H$.
\end{lemma}

\begin{proof}
In one direction, suppose $v$ is $d$-close to $L$. Then $v = u + \alpha_1 e_{i_1} + \alpha_2 e_{i_2} + \cdots + \alpha_{d'} e_{i_{d'}}$ for $u \in L$, elementary basis vectors $e_{i_1}, e_{i_2}, \ldots, e_{i_{d'}}$, scalars $\alpha_1, \ldots, \alpha_{d'} \in \F$ and $d' \le d$. Then $s = Hv = H(u + \alpha_1 e_{i_1} + \alpha_2 e_{i_2} + \cdots + \alpha_{d'} e_{i_{d'}}) = \sum_{j=1}^{d'} \alpha_j H_{i_j}$ where $H_{i_j}$ is the $i_j$-th column of $H$ (since $Hu=0$ as $u \in L$). 

Conversely, if $s = \sum_{j=1}^{d'} \alpha_j H_{i_j}$ for some $d'\le d$ and $i_1, \ldots, i_{d'}$, then $v - \sum_{j=1}^{d'} \alpha_j e_{i_j}$ is in $L$: indeed, $H(v - \sum_{j=1}^{d'} \alpha_j e_{i_j}) = s - \sum_{j=1}^{d'} \alpha_j H_{i_j} = 0$.
\end{proof}

As a corollary, finding a vector $v$ which is $d$-far from $L$ reduces to finding a vector $s$ that \emph{cannot} be expressed as a linear combination of $d-1$ columns of $H$.

\begin{corollary}
\label[corollary]{cor:reduction-to-syndrome}
Suppose there exists a deterministic polynomial-time algorithm that, given a matrix $H \in \F^{(n-k)\times n}$ of full row-rank, finds a vector $s \in \F^{n-k}$ such that $s$ cannot be written as a linear combination of $d-1$ columns of $H$. Then there exists an algorithm that finds a vector $v$ which is $d$-far from $\ker H$.
\end{corollary}

\begin{proof}
Let $L=\ker H$.
By solving a system of linear equations we can find a vector $v$ such that $Hv=s$. Note that a solution always exists since $\rank(H)=n-k$, so for every $s \in \F^{n-k}$, the linear system has a solution. By \cref{lem:syndrome-characterization}, $v$ is $d$-far from $L$.
\end{proof}

\section{The Remote Point Problem over the Rational Numbers}
\label[section]{sec:rpp-Q}

In this section, we observe a very simple polynomial-time algorithm with optimal parameters for the remote point problem over the field of rational numbers. The algorithm uses the reduction of \cref{sec:syndrome} in order to find, given a basis of a $k$-dimensional subspace $L \subseteq \Q^n$, a vector $v$ which is $(n-k)$-far from $L$.

Let $H \in \Q^{m \times n}$ be the parity-check matrix of $L$, where $m=n-k$. By clearing denominators, we may assume $H$ has integer entries. Further, clearing denominators increases the bit complexity of each entry of $H$ by at most a polynomial factor. Now, let $M=\max_{i,j}\set{|H_{i,j}|}$ be the maximum magnitude of an element in $H$, and define $K':=m!M^m$ and $K=2K'+1$. Finally, let $s \in \Q^{m}$ be defined as
\begin{equation}
\label{eq:syndrome}
s := (1,K,K^2, \ldots, K^{m-1}).
\end{equation}

Observe once again that the bit complexity of $K$ is polynomial in the bit complexities of the elements of $H$ (and hence polynomial in the input length), and thus the same is true for the entries of the vector $s$.

\begin{claim}
\label[claim]{cl:s-is-far}
For any $m' \le m-1$, $s$ cannot be written as a linear combination of $m'$ columns of $H$.
\end{claim}

\begin{proof}
Consider any set of $m' \le m-1$ columns of $H$. Clearly, we may assume the columns are linearly independent, and by adding more columns if necessary, we may assume the set has size exactly $m-1$. Let $A$ denote the resulting $m \times (m-1)$ submatrix of $H$.

We now show, by a Cramer's rule-type argument, that there exists a non-zero vector $a \in \Q^m$ with ``small'' entries such that $a^T A = 0$ but $a^T s \neq 0$, implying that $s$ is not in the image of the columns of $A$.

For $i \in [m]$, let $A_i$ denote the $(m-1) \times (m-1)$ submatrix of $A$ obtained by deleting the $i$-th row, and let $a_i = (-1)^i \det(A_i)$. Note that $|a_i| \le m!M^m = K'$. Define $a=(a_1, \ldots, a_m)$, and observe that $a$ is not the zero vector, since $\rank(A)=m-1$ so at least one of the $A_i$'s has full rank and $\det(A_i) \neq 0$.

We claim that $a^T A = 0$: indeed, if $c_j$ denotes the $j$-th column of $A$, then the fact that $a^T c_j = 0$ can be seen by considering the $m \times m$ matrix $A^{(j)}$ obtained by appending $c_j$ as the $m$-th column of $A$, and expanding the determinant of $A^{(j)}$ (which equals $0$, since the column $c_j$ appears in it twice) along the last column.

Finally, we show that $a^T s \neq 0$. Let $j_0$ be the largest index for which $a_{j_0}$ is non-zero. Then $|a_{j_0} s_{j_0}| \ge K^{j_0-1}$, whereas
\[
\left| \sum_{i=1}^{j_0-1} a_i s_i \right| \le \sum_{i=1}^{j_0-1} |a_i s_i| \le  K'(1+K+K^2 + \cdots + K^{j_0-2}) < K^{j_0-1}
\]
(since $K' < K$, and by the uniqueness of expansion in base $K$). It follows that
\[
a^T s = \sum_{i=1}^{j_0} a_i s_i \neq 0. \qedhere
\]
\end{proof}

\begin{corollary}
\label[corollary]{cor:rpp-over-Q}
There exists a deterministic polynomial-time algorithm that, given a basis of a $k$-dimensional subspace $L \subseteq \Q^n$, finds  a vector $v$ which is $(n-k)$-far from $L$.
\end{corollary}

\begin{proof}
The algorithm first constructs $s \in \Q^{n-k}$ as in \eqref{eq:syndrome} (as commented above, this can be done in polynomial time). By \cref{cl:s-is-far}, $s$ cannot be expressed as a linear combination of less than $n-k$ columns of $H$. The statement now follows from \cref{cor:reduction-to-syndrome}.
\end{proof}

\section{The Remote Point Problem over Finite Fields}
\label[section]{sec:rpp-finite}

In this section we obtain an improvement to the algorithm of Alon, Panigrahy and Yekhanin \cite{APY09}. As a part of their algorithm, they prove the following theorem:

\begin{theorem}[\cite{APY09}]
\label[theorem]{thm:base-algo}
    Let $L \subseteq \F^n$ be a linear subspace, and let $B_d$ denote the Hamming ball of radius $d$. There exists a collection $\calS$ of $\binom{t}{d}$ subspaces, each of dimension at most $\dim L + \frac{d}{t}n$, such that $L + B_d \subseteq \bigcup_{V\in \mathcal{S}} V$.
\end{theorem}

If $\dim L \le n/2$, $t=4d$ and $d=O(\log n)$, $\F^n$ contains a vector not in $\bigcup_{V\in \mathcal{S}} V$ which is found using a potential-function-based algorithm in \cite{APY09}, whose running time is polynomial in $n$ and $\binom{t}{d} = 2^{O(d)}$. When $d=O(\log n)$, this algorithm runs in polynomial time. A different and parallel algorithm was given by Arvind and Srinivasan \cite{AS10-epsbiased}, but they also use the structural result of \cref{thm:base-algo}. We remark that \cite{APY09} state their result for the field $\F_2$, but the proof can be easily modified for other finite fields, and was extended by Arvind and Srinivasan \cite{AS10-epsbiased} further to finite abelian groups.

We also use \cref{thm:base-algo} in order to prove:

\begin{theorem}
\label[theorem]{thm:small-improvement}
Let $\F$ be a fixed finite field and $L \subseteq \F^n$ be a linear subspace of dimension $k \le n/2$. There exists an algorithm that runs in time polynomial in $n$ and finds a point that is $\Omega(\frac{n}{\max\{k, \log n\}} \cdot \log n)$ far from $L$.
\end{theorem}

\begin{proof}
    Let $b=\max\{2k, 20 \log n\}$.
    Given a subspace $L$ of dimension $k$, partition the $n$ coordinates into $n/b$ blocks, denoted $B_1, \ldots, B_{n/b}$, of size $b$ each.
    The projection of $L$ on each block $B_i$ has dimension at most $k$.

    For every $i \in [n/b]$, use \cref{thm:base-algo} (with ambient dimension $b$) to find a point $v_i \in \F^{b}$ that is $d$-far from the projection of $L$ on $B_i$, for $d=\log n$ and $t=4 \log n$. This is possible, since the number of subspaces in the collection $\calS$ in \cref{thm:base-algo} is at most $\binom{t}{d} \le 2^{4 \log n} = n^4$ and each subspace has dimension at most
    \[
    k + \frac{d}{t}b \le \frac{b}{2} + \frac{b}{4} = \frac{3b}{4},
    \]
    where the inequality uses the fact that $k \le \frac{b}{2}$ and $\frac{d}{t}=\frac{1}{4}$.
    Hence, the total number of points covered by the union of subspaces in $\calS$ is at most 
    \[
    n^4 \cdot |\F|^{\frac{3b}{4}} < |\F|^{b},
    \]
    where we use the facts that $b \ge 20 \log n$ and $|\F| \ge 2$, so $n^4 < |\F|^{b/4}$.
    Hence, there exists a point that is not in the union of subspaces in $\calS$, which will be found by the algorithm from \cref{thm:base-algo}.
    
    The concatenation of $v_1, ..., v_{n/b}$ gives a point that is $\frac{n}{b} \cdot \log n$ far from $L$ since for every $u \in L$, the restriction of $u$ to $B_i$ differs from $v_i$ in at least $\log n$ coordinates.
\end{proof}

The difference comes from the fact that the analogous theorem in \cite{APY09} uses the same strategy but on every block finds a point that is $(\log k)$-far. The observation is that when $k$ is very small, we can run the algorithm from \cref{thm:base-algo} with improved remoteness parameters since we only want the running time to be polynomial in $n$.

\section{Open Problems}

This work raises several natural open problems.

\begin{itemize}
\item Can the algorithm from \cref{thm:rpp-over-Q} be adapted to obtain a deterministic algorithm for RPP over finite fields, for any remoteness parameter $d=\omega(\log n)$? Small fields like $\F_2$ seem like the hardest case; perhaps a more viable intermediate goal is to consider fields with size growing polynomially or even exponentially with $n$, noting that even when the field size is exponential in $n$, representing each element in it requires only polynomially many bits.
\item Can the algorithm from \cref{thm:rpp-over-Q} be adapted to obtain an explicit construction of rigid matrices over $\Q$? It is possible to obtain ``semi-explicit'' rigid matrices over $\Q$ in which the entries are integers of doubly-exponential magnitude (and hence exponential bit complexity) \cite{Lokam09, Ramya20}. The techniques for proving that such matrices are rigid are somewhat reminiscent of our proof of \cref{thm:rpp-over-Q}. Nevertheless, we are not able to obtain matrices whose entries are integers (or rationals) with polynomial bit complexity.

\item A natural intermediate goal is the construction of \emph{rigid sets}: these are sets that are guaranteed to contain, for every subspace $L$, a vector far from $L$. A rigid matrix corresponds to such a set of size $n$ (the set is the rows of the matrix), but one could relax that requirement and try to obtain sets that are as small as possible. \cite{APY09} construct such sets of size $2^{O(d)} n/d$ for remoteness parameter $d$ (see also \cite{AC15} for alternative constructions with roughly the same size). Constructions of size polynomial in $n$ for $d=\omega(\log n)$ are of great interest, even over $\Q$. They do not seem to imply any circuit lower bounds.

\end{itemize}

\section*{Acknowledgements}
The author thanks Mrinal Kumar and Ramprasad Saptharishi for various useful discussions on the remote point problem.
ChatGPT was used for help in proof organization and simplification, and proofreading. The paper was written by the human author.

\bibliographystyle{alphaurlpp}
\bibliography{references}

\end{document}